\documentclass{llncs}
\usepackage{amssymb,amsmath}
\usepackage{lineno}

\usepackage{nicefrac}
\usepackage{lipsum}
\usepackage{url}
\usepackage{color}
\usepackage[american]{babel}
\usepackage{graphicx}
\usepackage{version}
\usepackage{subcaption}
\usepackage{wrapfig}
\usepackage[textsize=footnotesize]{todonotes}

\newif\iffull
\fulltrue

\graphicspath{{figures/}}

\newcommand{\calC}{{\ensuremath{\cal C}}}
\newcommand{\calF}{{\ensuremath{\cal F}}}
\newcommand{\calB}{{\ensuremath{\cal B}}}
\newcommand{\calT}{{\ensuremath{\cal T}}}

\renewcommand{\int}{\mathit{int}}
\newcommand{\skel}{\mathit{skel}}
\newcommand{\leaveout}[1]{}

\newcommand{\odd}[2]{{\ensuremath{\mathit{odd}(#1\setminus #2)}}}

\date{}
\iffalse
\title{Finding big matchings in planar graphs quickly}
\author{Therese Biedl
\thanks{David R.~Cheriton School of Computer
Science, University of Waterloo, Waterloo, Ontario N2L 1A2, Canada.
Supported by NSERC}
}
\else
\title{Finding big matchings in planar graphs quickly}
\author{Therese Biedl
\thanks{Supported by NSERC. } 
}
\institute{David R.~Cheriton School of Computer
Science, University of Waterloo, Waterloo, Ontario N2L 1A2, Canada.
\email{biedl@uwaterloo.ca}
}
\fi

\begin{document}

\maketitle
\begin{abstract}
It is well-known that every $n$-vertex planar graph with minimum degree 3
has a matching of size at least $\frac{n}{3}$.  But all proofs of this 
use the Tutte-Berge-formula for the size of a maximum matching. 
Hence these proofs
are not directly algorithmic, and to find such a matching one must apply
a general-purposes maximum matching algorithm, which has run-time $O(n^{1.5}\alpha(n))$
for planar graphs.  In contrast to this, this paper gives a linear-time
algorithm that finds a matching of size at least $\frac{n}{3}$ 
in any planar graph with minimum degree 3.
\end{abstract}

\section{Introduction}

In 1979, Nishizeki and Baybars proved the following result 
on matchings in planar graphs (detailed definitions are in
Section~\ref{sec:def}):

\begin{theorem}\cite{NB79}
\label{thm:main}
Every $n$-vertex simple planar graph with minimum degree 3 has a matching of size
at least $\frac{n}{3}$.
\end{theorem}

Their proof relies on the famous Tutte-Berge formula \cite{Berge1958}, which states
that the number of unmatched vertices in a maximum matching equals
the minimum of $\odd{G}{T}-|T|$, where the minimum is taken over
all vertex sets $T$, and $\odd{G}{T}$ is the number of components of
odd size in the graph $G\setminus T$.  Nishizeki and Baybars then argued,
by considering the planar bipartite graph of the edges between $T$ and $V-T$,
that $\odd{G}{T}-|T|\leq \frac{n}{3}$ for any $T$. This implies that the maximum matching has
size at least $\frac{n}{3}$.%
\footnote{They actually proved a bound of $\frac{n+2}{3}$, and generally
bounds of the form $\frac{n}{3}+c$
for some small constant $c$, where $c$ depends on the connectivity of the
graph and on $n$ being sufficiently big.  To keep the statement of 
results simpler, such constants will be omitted except where they are crucial
to make an induction work.}

For planar 3-connected graphs (which always have minimum degree 3), a second
independent proof of the matching-bound of $\frac{n}{3}$ was given in 2004 \cite{BDD+04}.
The proof again uses the Tutte-Berge formula, but
the approach to prove that $\odd{G}{T}-|T|\leq \frac{n}{3}$  is different:  
study the graph induced by $T$, and count 
the faces that have odd components inside. 

Both arguments have the drawback that they 
provide no insights for finding a matching of size at least $\frac{n}{3}$ efficiently.
One can find such a matching by using any algorithm that finds a maximum matching
in a graph. The first such algorithm was Edmond's famous blossom algorithm
\cite{Edmonds1965a}, which finds a maximum matching by repeatedly finding an
augmenting path to increase the matching by one edge.  This takes $O(m)$ time
per augmenting path, hence $O(mn)$ time overall, where $m$ is the 
number of edges of the graph.  The run-time was improved
by Micali and Vazirani \cite{MV80} (with corrections by Vazirani \cite{Vaz12}) 
to be $O(\sqrt{n}m\alpha(m,n))$ time, where $\alpha(m,n)$ is the inverse 
Ackerman function that grows very slowly, For planar graphs
this amounts to $O(n^{1.5}\alpha(n,n))$, which is not linear time.  
One can find an
$\varepsilon$-approximation of the  maximum matching in run-time $O(8^{1/\varepsilon} n)$
\cite{Baker94}, and there are various fast matching-algorithms for 
planar graphs 
in special situations (see e.g.~\cite{BBD+01,BKM+17,AV18} and the
references therein).  But to the author's knowledge there exists no prior algorithm
that can find the matching of Theorem~\ref{thm:main} in linear time.

This paper gives another proof for Theorem~\ref{thm:main}.
In contrast to the previous proofs, it does not use the Tutte-Berge formula.
Instead, for a 3-connected planar graph it uses a spanning tree of maximum degree 3,
which is known to exist and can be found in linear time.    Then it finds a
maximum matching in this spanning tree.
To find a large matching in graphs that are not 3-connected, split
it into components of higher connectivity, find matchings in them, and 
combine suitably.  (However, there are some pitfalls to this if the components
contain too few vertices.)
This proof naturally gives rise to an algorithm that runs in linear time.

\section{Definitions}
\label{sec:def}

Let $G=(V,E)$ be a graph with vertices $V$ and edges $E$, and
use $n:=|G|:=|V|$ as convenient shortcuts.
The {\em degree} $\deg_G(v)$ is the number of incident edges of
vertex $v$ in graph $G$; write $\deg(v)$ if the graph in question
is clear.  Graph $G$ has {\em minimum
degree $\delta$} if every vertex has at least $\delta$ incident edges.
$G$ is called {\em $k$-connected}
if removing any $k-1$ vertices leaves a connected graph.  
A {\em matching} $M$ in $G$ is a set of edges
that have no endpoint in common.  Call a vertex {\em matched} by
$M$ if it is an endpoint of an edge of $M$, and {\em unmatched} otherwise.
This paper frequently uses matchings where some vertices
are required to be unmatched.  Specifically, if three vertices
$\{x,y,z\}$ are fixed, then 
$M_A$ (for any set $A\subset \{x,y,z\}$) denotes a matching where
$\{x,y,z\}\setminus A$ must be unmatched.  Put differently, vertices
in $A$ are allowed to be used by the matching, while vertices
in $\{x,y,z\}\setminus A$ are forbidden to be used.

Throughout this paper the input graph $G$ is assumed to be {\em simple}, i.e., 
to have neither loops nor parallel edges.  (There are no non-trivial
matching-bounds for graphs of large minimum degree if loops or parallel edges 
are allowed.)

Throughout the paper the input graph $G$ is assumed to be
{\em planar}, i.e., can be drawn without a
crossing in the plane, and that one particular such drawing has
been fixed.    All subgraphs of $G$ are assumed to inherit this drawing.
A planar drawing $\Gamma$ defines the {\em faces}, which are the regions of
$\Bbb{R}^2\setminus \Gamma$ and are characterized by giving the
vertices and edges that are on its boundary.  The
{\em outer-face} is the infinite region of $\Bbb{R}^2\setminus \Gamma$.
An {\em interior vertex} is a vertex not on the outer-face.
A planar graph is called {\em triangulated} if all faces have three
vertices.

\section{Finding a matching}

The algorithm for Theorem~\ref{thm:main} proceeds by splitting the
graph into subgraphs of increasingly higher connectivity, finding
matchings in the subgraphs, and putting them together.  To ensure
that the combined matchings to not share an endpoint, 
some vertices on the outer-face are required to be unmatched.

\subsection{3-connected graphs}
\label{subse:3conn}

For 3-connected graphs, a matching that satisfies the bound of
Theorem~\ref{thm:main} can be found by putting some existing
tools together in a suitable way.

\begin{lemma}
\label{lem:main_3conn}
Let $G$ be a 3-connected planar graph where the outer-face contains a vertex $x$
and an edge $(y,z)$ with $y\neq x\neq z$.  Then $G$ has 
matchings of the following sizes:
\begin{enumerate}
\item \label{it:Mx}
a matching $M_{x}$ of size at least $\frac{n-3}{3}$ where $y,z$ are 
	unmatched, and
\item \label{it:Mempty} a matching $M_{\emptyset}$ of size at least $\frac{n-4}{3}$ where $x,y,z$
	are unmatched.
\end{enumerate}
Furthermore, these matchings can be found in linear time.
\end{lemma}
\begin{proof}
Crucial for the proof is the concept of a {\em binary spanning tree}, which is a spanning tree 
$T$ such that $\deg_T(v)\leq 3$  for all vertices $v$.  (This is also known
in the literature as a {\em 3-tree}, but the term `binary spanning tree' is preferable
since `3-tree' is also used for maximum graphs of treewidth 3.)

Barnette \cite{Bar66} showed that every 3-connected planar graph has a
binary spanning tree, and it can be found---using a variety of methods---in
linear time \cite{Strothmann-PhD,Bie-Barnette,BiedlKindermann}.  Some
of these binary spanning trees have additional properties, and the one
most useful for the proof is the one found with the algorithm by the author
\cite{Bie-Barnette}. This binary spanning is derived
from a so-called canonical ordering $\langle v_1,v_2,\dots,v_n\rangle$ of
the vertices, and has the property that $\deg_T(v_1)=\deg_T(v_n)=1$, 
$\deg_T(v_2)=2$, and edge $(v_1,v_2)\in T$.    Since the
canonical ordering can be chosen such that $v_1=y$, $v_2=z$ and $v_n=x$ \cite{Kant96},
therefore $\deg_T(y)=1=\deg(x)$ and $\deg_T(z)=2$ and $(y,z)\in T$.
This implies that $T_\emptyset := T\setminus \{x,y,z\}$ is also a tree.
See Figure~\ref{fig:3tree}(a).

\begin{figure}[ht]
\hspace*{\fill}
\subcaptionbox{}{\includegraphics[width=0.6\linewidth,page=1]{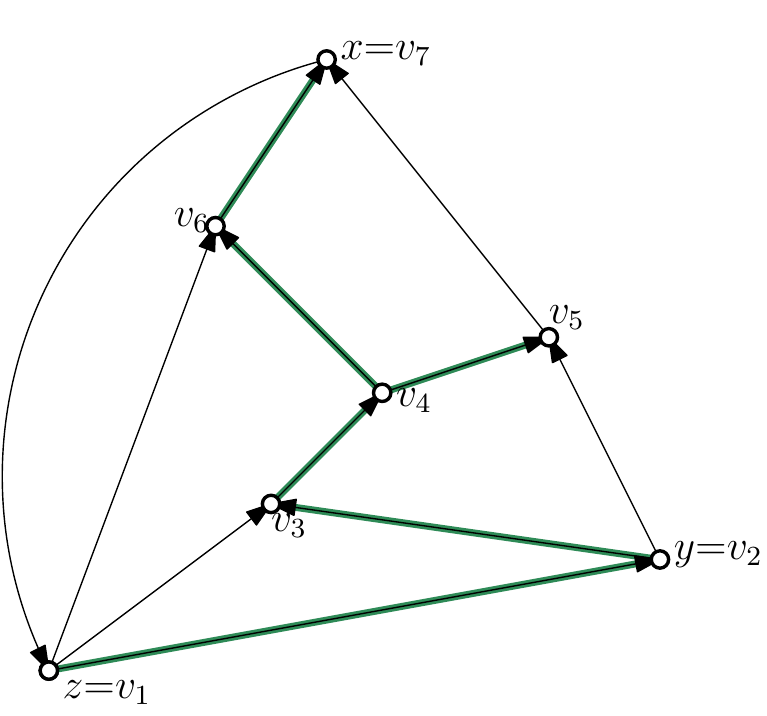}}
\hspace*{\fill}
\begin{minipage}[b]{0.25\linewidth}
\subcaptionbox{}{\includegraphics[width=0.992\linewidth,page=3]{spanningBinaryTree.pdf}}

\medskip
\subcaptionbox{}{\includegraphics[width=0.99\textwidth,page=4,trim=0 0 0 0,clip]{spanningBinaryTree.pdf}}
\end{minipage}
\hspace*{\fill}
\caption{(a) A planar 3-connected graph with a canonical ordering \cite{Kant96}
and the binary spanning tree obtained from it with the method in \cite{Bie-Barnette}.  
Any matching $M_\emptyset$ in it has size $\leq \frac{n-4}{3}$.
(b) The matching-bound $|M_x|\geq \frac{n-2}{x}$ is tight.
(c) A planar 3-connected graph with $n=6$ vertices has a matching $M_x$
of size 2.
}
\label{fig:3tree}
\end{figure}

The second ingredient to the proof is the result that every connected graph with
minimum degree 3 has a matching of size at least $\frac{n-1}{3}$ 
\cite{BDD+04}.  It is not known how to find such a matching in time faster
that any general-purpose maximum matching algorithm in arbitrary graphs, but 
the result is used here only in tree $T_\emptyset$.  Tree $T_\emptyset$ is connected 
and has maximum degree 3, and so it contains a matching $M_\emptyset$ of size at least
$\frac{n(T_\emptyset)-1}{3} = \frac{n-4}{3}$.  Also, in a tree a maximum
matching can be found in linear time with a straightforward dynamic programming
algorithm (see also \cite{Baker94}), so $M_\emptyset$ can be found in linear time.

This proves (\ref{it:Mempty}), and (\ref{it:Mx}) is proved in the same way by using tree
$T_{x}:=T\setminus \{y,z\}$ instead.
\qed
\end{proof}


The bound of Lemma~\ref{lem:main_3conn}(\ref{it:Mempty}) is tight, see
for example Figure~\ref{fig:3tree}(a).  The bound for
Lemma~\ref{lem:main_3conn}(\ref{it:Mx}), on the other hand, can be improved
slightly, to $\frac{n-2}{3}$.  Such a minor improvement could be deemed
not be worth the effort, but is of {\em vital} importance when it
it comes to merging 3-connected components of a 2-connected graph

\begin{lemma}
\label{lem:improvement}
Let $G$ be a 3-connected
planar graph where the outer-face contains a vertex $x$
and an edge $(y,z)$ with $y\neq x\neq z$.  If $n\geq 4$,  then $G$ has 
a matching $M_{x}$ of size at least $\frac{n-2}{3}$ where $y$ and $z$ are
unmatched.
It can be found in linear time.
\end{lemma}
\begin{proof}  
Note that it suffices to prove the {\em existence} of such a matching;
to find it in linear time use matching $M_x'$ of size at least $\frac{n-3}{3}$
that exists by Lemma~\ref{lem:main_3conn}(\ref{it:Mx}) and then (if needed) run one round of
Edmond's blossom algorithm to increase its size by 1.  This gives the matching
in linear time.

To show that matching $M_x$ exists, consider cases.  First, if $n\not\equiv 3\bmod 3$,
then by integrality $|M_x'|\geq \lceil \frac{n-3}{3} \rceil \geq \frac{n-2}{3}$ and
the result holds.   Second, if $n\geq 8$, then it is known \cite{BDD+04} that $G$
has a matching $M$ of size at least $\frac{n+4}{3}$.  Removing from $M$ the (at most
two) edges that are incident to $y$ and $z$ hence gives a matching of size
at least $\frac{n-2}{3}$.  By $n\geq 4$ this leaves only the case $n=6$.
One could now obtain the result by enumerating all possible 3-connected
graphs on $6$ vertices, but the following gives an explicit proof
(see Figure~\ref{fig:3tree}(c)).

Let $w_1,w_2,w_3$ be the three vertices of $G$ that are not $\{x,y,z\}$.
By Lemma~\ref{lem:main_3conn}(\ref{it:Mempty}) there is a matching
$M_\emptyset$ of size 1 among $\{w_1,w_2,w_3\}$, after renaming 
therefore $(w_1,w_2)$ is an edge.
If $(x,w_3)$ is an edge then $(w_1,w_2)\cup (w_3,x)$
is the desired matching, so assume that $(x,w_3)\not\in E$.  Since
the graph is 3-connected and $n=6$, vertex $x$ has
$\deg(x)\geq 3$, so $x$ is adjacent to at least one of $w_1,w_2$. After
renaming therefore assume that $(w_1,x)$ is an edge.
If $(w_2,w_3)\in E$, then $(x,w_1)\cup (w_2,w_3)$ is the desired matching, so 
assume $(w_2,w_3)\not\in E$.  By $\deg(w_3)\geq 3$ therefore $w_3$ must be
adjacent to $w_1,y,z$.  If $(x,w_2)$ is an edge, then $(x,w_2)\cup (w_1,w_3)$
is the desired matching, so assume $(x,w_2)\not\in E$.  By minimum degree 3
(and since neither $x$ nor $w_2$ is adjacent to $w_3$) therefore
$x$ and $w_2$ must both be adjacent to all of $w_1,y,z$.  But now there is
a complete bipartite graph $K_{3,3}$ with sides $\{w_1,y,z\}$ and $\{x,w_2,w_3\}$
in the planar graph $G$, a contradiction.
\qed
\end{proof}

The reader may be a bit disappointed by this proof, because it 
uses the Tutte-Berge formula (by appealing to the
result of Nishizeki and Baybars \cite{NB79}) and relies on Edmonds'
blossom algorithm \cite{Edmonds1965a}, if only for one round.
Both can be avoided by inspecting how
binary spanning trees can be obtained and finding one extra edge
via a counting argument.  (Details are omitted, but the appendix
has a similar proof for graphs with minimum degree 4.) Unfortunately
the resulting linear-time algorithm, while avoiding the blossom algorithm,
is also not simple since it needs to compute a so-called Tutte path.

\subsection{2-connected graphs}

For 2-connected graphs, it suffices to exclude two vertices on
the outer-face.

\begin{lemma}
\label{lem:main_2conn}
Let $G$ be a 2-connected planar graph with minimum degree 3
where the outer-face contains a vertex $x$
and an edge $(y,z)$ with $y\neq x\neq z$.  Then $G$ has 
a matching $M_{x}$ of size at least $\frac{n-2}{3}$ where $y,z$ are unmatched.
It can be found in linear time.
\end{lemma}
\begin{proof}
The claim holds by Lemma~\ref{lem:improvement} if $G$ is 3-connected,
so assume that it is not.  The idea to find a $M_x$ is
fairly standard: compute the 3-connected components, find a large 
matching in them, and then merge the matchings.  However, two details 
are unusual.  First, the so-called S-nodes will create
difficulties because their matchings are not big enough.  Therefore
a pre-processing step (illustrated in Figure~\ref{fig:SPQR})
modifies the graph by inserting {\em artificial}
edges such that there are no S-nodes.  This creates a complication,
because the algorithm to find matching $M_x$ later must ensure 
that $M_x$ does not use artificial edges.  Second,
the matchings in the 3-connected components might well use so-called
virtual edges that do not exist in $G$; 
such virtual edges must be removed from the matchings and compensated
for by using bigger matchings in adjacent 3-connected components.
(This necessitates parsing 3-connected components in a particular
order, rather than simply putting the matchings together.)

\paragraph{SPQR-trees: }
Recall first the definition of 3-connected components and the
associated {\em SPQR-tree}, see for example \cite{DiBattista96b,GM00}.
Let $G$ be a graph that that is not 3-connected, and let $(u,v)$
be a {\em cutting pair}, i.e., $G-\{u,v\}$ splits into (possibly
many) connected components $C_1,\dots,C_k$.  The {\em cut-components}
of $(u,v)$ are the subgraphs $C_1^+,\dots,C_k^+$ obtained by taking
for $i=1,\dots,k$ the subgraph $C_i$ and adding to it the vertices $u,v$
and all edges from them to $C_i$.  Also add a {\em virtual edge} $(u,v)$
to subgraph $C_i^+$.  Finally, if $k\geq 3$ or if edge $(u,v)$ exists
in $G$, then define one more cut-component to consist of vertices $(u,v)$
with $k$ virtual edges $(u,v)$ (as well as the edge $(u,v)$ if it existed
in $G$).  Each cut-component that has at least four vertices is again
2-connected; repeat the process within them.     

The {\em 3-connected
components} of graph $G$ are the graphs obtained by repeating this process until
all resulting graphs are multiple edges or triangles or 3-connected simple planar graphs
with at least four vertices.  For each 3-connected component $C$, create a node $\mu$ and set
$\skel(\mu)=C$; note that $\skel(\mu)$ includes virtual edges.
Classify each node $\mu$ as follows:  $\mu$ is a {\em P-node} if
$\skel(\mu)$ has two vertices and at least three parallel edges, 
it is an {\em S-node} if $\skel(\mu)$ is a triangle,
and it is an {\em R-node} if $\skel(\mu)$ 
is a 3-connected simple graph with at least 4 vertices.%
\footnote{Sometimes S-nodes are combined to represent longer cycles.
This is not done here since
S-nodes need to be handled in a special way anyway to find big matchings.
In consequence the SPQR-tree as defined in this paper is not necessarily
unique; any SPQR-tree will serve to find the matching.}  Note that any
virtual edge was created as two copies, hence exists in two 3-connected components;
connect two nodes $\mu,\nu$ if they have a virtual edge in common.
One can argue that the result is a tree $\calT$, called {\em SPQR-tree}.  The
SPQR-tree and the 3-connected components can be found in linear time \cite{GM00}. 

Root the SPQR-tree $\calT$
at the node $\rho$ that contains the edge $(y,z)$ and set
this edge to be the {\em parent-edge} of root-node $\rho$.  For
all other nodes $\mu$ in $\calT$, the {\em parent-edge} $(u_\mu,v_\mu)$
is the unique virtual edge that $\mu$ has in common with its parent in $\calT$.

\paragraph{Eliminating S-nodes:}
As explained above, any S-node $\mu$ has to be eliminated by inserting artificial
edges, because the process later inserts a matching of $\skel(\mu)$ that does
not use the ends of the parent-edge, and in an S-node this would add no
matching-edge at all.

Process S-nodes in top-down order in $\calT$.  Let $\mu$ be an S-node that
has not yet been processed, and let 
$(u_\mu,v_\mu)$ be its parent edge and $a$ be the node $\neq u_\mu,v_\mu$
in the triangle $\skel(\mu)$.
At least one of the edges $(u_\mu,a)$ and $(v_\mu,a)$ is virtual,
else $\deg_G(a)=2$, contradicting that the minimum
degree is 3.%
\footnote{This is the only place where minimum degree 3 is used for
2-connected graphs; the result would hence also hold for any 2-connected
planar graph where any S-node of the SPQR-tree is adjacent to at least
two other nodes of the SPQR-tree. } 
    After possible renaming, edge $(v_\mu,a)$ is virtual in $\skel(\mu)$.
Let $\nu'$ be the child of $\mu$ that also contains virtual edge $(v_\mu,a)$.
Set $\nu=\nu'$ if $\nu'$ is an R-node or an S-node, and set $\nu$ to be
a child of $\nu'$ if $\nu'$ is a P-node.  (This exists since P-nodes
are never leaves of $\calT$.)  Since no two P-nodes are adjacent, 
regardless of the case now $\nu$ is not a P-node and the parent-edge
of $\nu$ is $(v_\mu,a)$.  Furthermore $\skel(\nu)$ is simple and 
has at least three vertices.
Consider one face  in $\skel(\nu)$
that is incident to its parent-edge $(v_\mu,a)$, and let $b\neq v_\mu,a$
be a vertex of $\skel(\nu)$ on this face.  

\begin{figure}[ht]
\hspace*{\fill}
\includegraphics[scale=0.7,page=1]{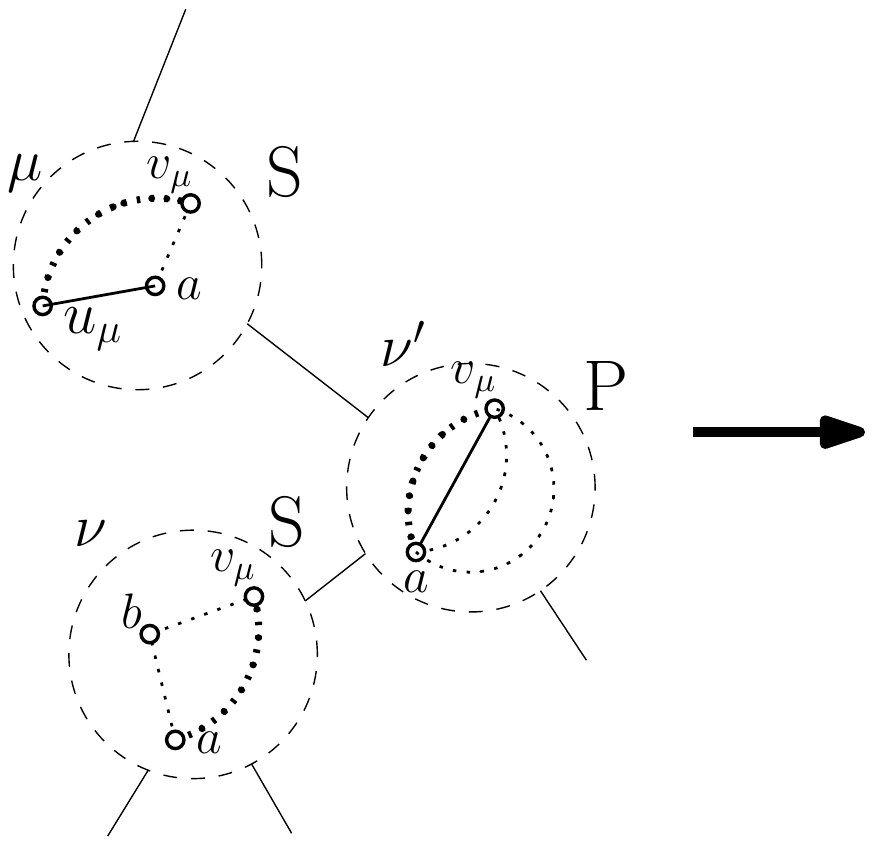}
\includegraphics[scale=0.7,page=2,trim=0 0 80 0,clip]{SPQR.pdf}
\hspace*{\fill}
\caption{Eliminating an S-node $\mu$ by inserting an edge that combines
$\mu$ with a child or grandchild $\nu$ to form an R-node.  Parent-edges
are thick dotted, other virtual edges are thin dotted, the artificial
edge is red (beaded).}
\label{fig:SPQR}
\end{figure}

Add an {\em artificial edge} $(u_\mu,b)$ to the graph.  
Note that one can embed
$\skel(\nu)$ inside $\skel(\mu)$ such that $b$ is on
a face incident to $(u_\mu,v_\mu)$, therefore the
resulting graph is planar.  (This requires possibly
permuting the order of virtual edges on P-nodes between
$\mu$ and $\nu$ and/or reversing the embedding of $\nu$.) 
Note that with this $\mu$ and $\nu$ becomes one R-node,
whose children are the remaining children
(if any) of $\mu$, all children of $\nu$, and (if $\nu'$ was
a P-node) the remaining children of $\nu'$ (connected via a P-node
if $\nu'$ had three or more children).
These changes are local to 
$\calT$, and therefore can be executed in constant time per S-node.

Observe that if an R-node $\mu'$ in the
resulting SPQR-tree has an artificial edge $e$ in $\skel(\mu')$,
then one end of $e$ belongs to the parent-edge of $\mu'$
while the other is {\em not} on the parent-edge.  In particular,
$\skel(\mu')$ has 
a planar embedding such that the parent-edge is
on the outer-face while the other end of $e$ is interior;
this will be important later to avoid adding an artificial
edge to the matching.

\paragraph{Finding the matching:}
Let $G^+$ be the graph obtained after all S-nodes have been
eliminated by the insertion of artificial edges.  Let
$\calT^+$ be its SPQR-tree, with the root-node $\rho$
containing $(y,z)$.  It will be helpful to expand $\calT^+$ with a new
root-node $\rho'$ that stores two copies of $(y,z)$ (one real and one
virtual) and to make $(y,z)$ virtual in $\rho$ (which becomes the child
of $\rho'$).

Now 
compute a matching $M_x$ while {\em re-assembling} $G^+$ from $\calT^+$.
This means the following process: 
Initially let the {\em parsed} subgraph $G_p$ be $\skel(\rho')$, i.e.,
the double edge $(y,z)$.
Then, while $G_p$ contains virtual edges,
take one virtual edge $(u,v)$ and let $\mu$ be the node of $\calT^+$
for which this is the parent-edge.
Add to $G_p$ the graph $\skel(\mu)$ , and delete from $G_p$
the two copies of the virtual edge $(u,v)$.   During this process the
goal is to maintain a matching $M_x$ of $G_p$
that satisfies the following:
\begin{itemize}
\item $|M_x|\geq \frac{|G_p|-2}{3}$.
\item Vertices $y$ and $z$ are unmatched in $M_x$.
\item $M_x$ uses no artificial edge.
\item $M_x$ uses only virtual edges for which at least one cut-component
	has not yet been re-inserted.
\end{itemize}

For the initial $G_p$ (consisting of $y,z$ and two edges between them),
set $M_x$ to be the empty set and observe that all conditions hold.
Now consider re-inserting $\skel(\mu)$ for some node
$\mu$ of $\calT^+$, and let $(u_\mu,v_\mu)$ be its parent-edge.
All ancestors of $\mu$ have already been re-inserted, and in particular 
vertices $u_\mu$ and $v_\mu$ are parsed already.  Furthermore, $\mu$ is a P-node or an
R-node since S-nodes have been eliminated.  Expand $M_x$ with 
a matching $M_\mu$ inside $\skel(\mu)$ that is found in one of three
possible ways:

\medskip\noindent{\bf Case 1:} $\mu$ is a P-node.

See node $\mu_1$ in Figure~\ref{fig:SPQR_reassemble}.
In this case, do not add to matching $M_x$ at all.  Note that
no new vertices were parsed, so again $|M_x|\geq \frac{|G_p|-2}{3}$.
The removal of the virtual edge $(u_\mu,v_\mu)$ of $\mu$ is not
a problem, even if it was used by $M_x$, because this edge is
the parent-edge of at least one child $\nu$ of $\mu$ (since $\mu$
is a P-node), and $\skel(\nu)$ has not yet been re-inserted.

\medskip\noindent{\bf Case 2:} $\mu$ is an R-node, and $(u_\mu,v_\mu)$ was not used by $M_x$ and/or exists in $G$.

See node $\mu_2$ in Figure~\ref{fig:SPQR_reassemble}.
To find matching $M_\mu$, let $x_\mu\neq u_\mu,v_\mu$ be an arbitrary vertex
on the outer-face of $\skel(\mu)$; this exists since $\mu$ is an R-node.
Using Lemma~\ref{lem:improvement}, find a matching $M_\mu:=
M_{x_\mu}(\skel(\mu))$ that leaves $u_\mu,v_\mu$ unmatched and that
has size at least $\frac{|\skel(\mu)|-2}{3}$.  

Since $M_\mu$ leaves the ends of the parent-edge unmatched, it does not
use the artificial edge that may exist inside $\skel(\mu)$, and it does
not use an edge incident to $y$ or $z$ (since
these belong to the root-node, they are either not in
$\skel(\mu)$ or are an end of the parent-edge of $\mu$).
All endpoints of edges in $M_\mu$ 
are in $\skel(\mu)\setminus \{u_\mu,v_\mu)$, hence newly
parsed, hence not matched by $M_x$.  So $M_x\cup M_\mu$
is a matching of the required size.  It does not use
the virtual edge $(u_\mu,v_\mu)$ by case assumption, so it
satisfies all required conditions.

\medskip\noindent{\bf Case 3:} $\mu$ is an R-node, and $(u_\mu,v_\mu)\in M_x$ is not an edge of $G$.

See node $\mu_3$ in Figure~\ref{fig:SPQR_reassemble}.
Since the virtual edge $(u_\mu,v_\mu)$ will be removed from the current graph,
it must also be removed from $M_x$;  therefore  matching
$M_\mu$ in $\skel(\mu)$ must compensate for this removed edge and must be bigger 
than in the previous case.  
Note first that $(u_\mu,v_\mu)\in M_x$ implies
that neither of its endpoints is $y$ or $z$, 
so $y,z$ are not vertices of $\skel(\mu)$
and the second condition on $M_x$ holds. 
There may or may not be an artificial edge in $\skel(\mu)$;
if there is one then exactly one of its ends is in $\{u_\mu,v_\mu\}$.  Set
$y_\mu\in \{u_\mu,v_\mu\}$ to be this end (if there is an artificial edge)
or an arbitrary vertex in $\{u_\mu,v_\mu\}$ otherwise, and set $x_\mu$ to
be the other vertex in $\{u_\mu,v_\mu\}$.

Let $z_\mu$ be the neighbour of $y_\mu$ that is distinct from $x_\mu$ 
and on the outer-face of $\skel(\mu)$; this exists since $\mu$ is a R-node.
Set $M_\mu$ to be the matching $M_{x_\mu}(\skel(\mu))$ with respect
to these three vertices, i.e., $y_\mu$ and $z_\mu$ are unmatched in
$M_\mu$.  (In particular, the artificial edge has not been used.)
The artificial edge (if any) is also not edge $(y_\mu,z_\mu)$, since
one of its ends is interior to $\skel(\mu)$.  Now update
$M_x' := M_x \setminus (u_\mu,v_\mu) \cup M_\mu \cup (y_\mu,z_\mu)$
and verify all claims.  In particular, the size is at least
$\frac{|G_p|p-2}{3} - 1  + \frac{|\skel(\mu)|-2}{3} + 1$, which is
sufficiently large since $|\skel(\mu)|-2$ new vertices were parsed.

\begin{figure}[ht]
\hspace*{\fill}
\includegraphics[scale=0.56,page=1]{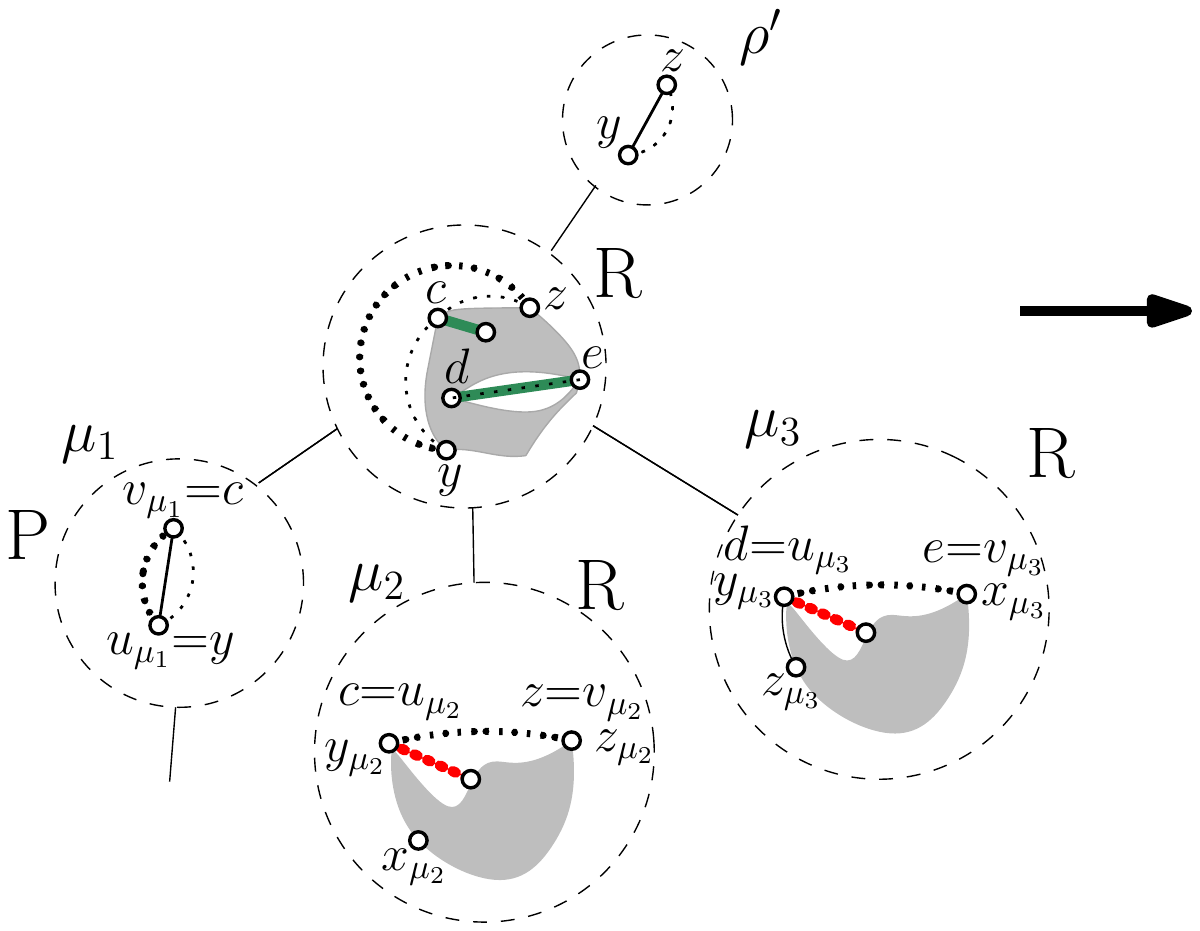}
\includegraphics[scale=0.56,page=2,trim=40 0 40 0,clip]{SPQR_reassemble.pdf}
\hspace*{\fill}
\caption{Re-assembling the graph while computing
a matching.}
\label{fig:SPQR_reassemble}
\end{figure}

Repeat the process until $G^+$ has been completely re-assembled.
The resulting matching $M_x^+$ is a matching of size at least $\frac{n-2}{3}$ 
where $y,z$ are unmatched and no artificial edges are
used.
Since the 3-connected components can be found in linear time,
and matching $M_x$ is obtained by computing a matching in
each 3-connected component, the overall run-time is linear.
\qed
\end{proof}

\subsection{Proof for connected graphs}

Theorem~\ref{thm:main} can now finally be proved by 
arguing how to handle cut-vertices;
this uses more or less the ``standard'' approach of splitting a graph into
2-connected components and putting matchings together, but again 
some artificial edges need to be inserted to ensure that 2-connected components
are sufficiently big.

\begin{proof} (of Theorem~\ref{thm:main})
Let $\calB$ be the {\em block-tree} of $G$, i.e., create a node 
for every {\em 2-connected component} (the maximal 2-connected subgraphs)
and a node for every {\em cut-vertex} (a vertex $v$ such that $G\setminus v$ is
not connected), and add
an edge from a cut-vertex $v$ to a 2-connected component $C$ if and only if
$v\in C$.   As the name suggests, $\calB$ is a tree, and it can be
computed in linear time \cite{Hopcroft73a}. Let $C_R$ be a leaf of $\calB$;
this is necessarily a 2-connected component with at least
four vertices by minimum degree 3.  Root the block-tree at $C_R$.
Let $(y,z)$ be an edge in $C_R$, and
after possible change of embedding, assume that $(y,z)$ is on the
outer-face of $G$.  
Declare $y$ to be the {\em parent-node} of $C_R$; for any 2-connected component
$C\neq C_R$ the {\em parent-node} is the cut-vertex stored at the parent of $C$ in the
block-tree.

\paragraph{Eliminating bridges:}
A {\em bridge} is a 2-connected component that contains a single edge $(u,v)$.
A pre-processing step (illustrated in Figure~\ref{fig:blocktree}(a-b))
eliminates these in top-down order in the block-tree $\calB$;
Let $C$ be a 2-connected component that contains only a bridge $(u,v)$.
By minimum degree 3 both $u$ and $v$ are cut-vertices, so after possible 
renaming $u$ is the parent-node of $C$ and
$v$ is a child of $C$. Let $C'$ be a child of $v$ and let $w$
be a neighbour of $v$; one can embed $C'$ such that $(v,w)$ is
on the outer-face of $C'$.  Insert an artificial edge $(u,w)$ and
note that the result is again planar.  Furthermore, this combines
$C$ and $C'$ into one 2-connected component whose parent-node is
again $u$; it inherits all
children of $C'$, and also inherits $v$ as child if there are
any other 2-connected components containing $v$.  This is again
a local change to the block-tree that can be done in overall linear time.
For future reference, note that any artificial edge
in a 2-connected component $C$ is adjacent to the parent-node
of $C$.

\begin{figure}[ht]
\hspace*{\fill}
\subcaptionbox{}{\includegraphics[scale=0.5,page=1]{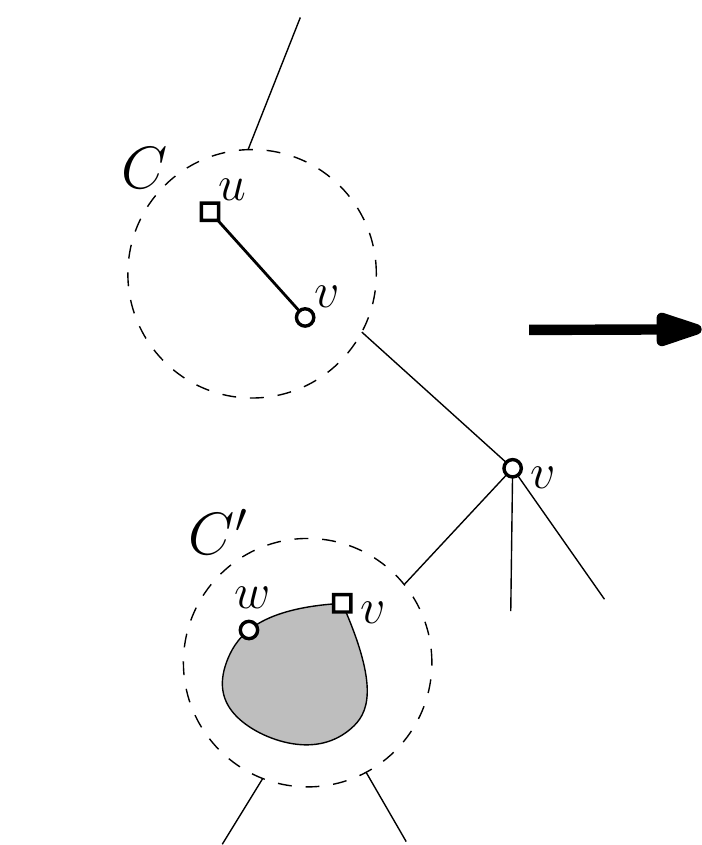}}
\subcaptionbox{}{\includegraphics[scale=0.5,page=2,trim=10 0 30 0,clip]{blocktree.pdf}}
\hspace*{\fill}
\subcaptionbox{}{\includegraphics[scale=0.5,page=3,trim=0 0 60 0,clip]{blocktree.pdf}}
\hspace*{\fill}
\caption{(a-b) Eliminating a bridge $C$ by inserting an edge that combines $C$
with a grandchild $C'$.  Parent-node are denoted by squares. (c) Finding a matching.}
\label{fig:blocktree}
\end{figure}

\paragraph{Finding the matching:}
Now, similar as for 2-connected graphs, re-assemble the graph 
from the 2-connected components while creating a matching.  
During this process compute a matching $M_x$ that satisfies the
following:
\begin{itemize}
\item $|M_x|\geq \frac{|G_p|-2}{3}$, where $G_p$ is the graph that
	has been parsed so far.
\item Vertices $y$ and $z$ are unmatched in $M_x$.
\item $M_x$ uses no artificial edge.
\end{itemize}

Initialize $M_x$ with
a matching in $C_R$ that leaves $y$ and $z$ unmatched.
Since $C_R$ has at least four vertices, 
this has size at least $\frac{|C_R|-2}{3}$ by Lemma~\ref{lem:main_2conn}.
This uses no artificial edge since such an edge is necessarily adjacent
to $y$ or $z$.  So all conditions are satisfied.

Now consider some 2-connected component $C$, and let $y_C$
be its parent-node;  $|C|\geq 3$ since there
are no bridges.  All ancestors of $C$ in $\calB^+$ have already
been re-inserted, and in particular $y_C$ has been parsed.
Expand $M_x$ with a matching $M_C$ inside $C$ that is 
found as follows
(see also Figure~\ref{fig:blocktree}(c)):
Temporarily add a vertex $z_C$ in the outer-face of $C$,
and connect it to $y_C$ and to two other vertices on
the outer-face of $C$ in such a way that edge $(y_C,z_C)$
is on the outer-face of the resulting graph $C^+$.
Let $x_C\neq y_C,z_C$ be an arbitrary vertex on the
outer-face of $C^+$.  Since $|C^+|\geq 4$ and $C^+$
is 2-connected, use Lemma~\ref{lem:main_2conn}  to find a matching $M_C:=M_{x_C}(C^+)$
of $C^+$ where $y_C$ and $z_C$ are unmatched such that
$|M_C|\geq \frac{|C^+|-2}{3} = \frac{|C|-1}{3}$.
All vertices matched by $M_C$ are in $C\setminus \{y_C,z_C\}$, hence newly
parsed, hence not matched by $M_x$.
Since $|C|-1$ vertices were newly parsed, 
therefore $M_x':=M_x\cup M_C$ is sufficiently big.  
Finally, since $y_C$ is unmatched in $M_x$, 
the artificial edge (if any)
in $C$ is not used by $M_C$.  So all conditions hold.

Repeat the process until $G^+$ has been completely re-assembled.
The resulting matching $M_x$ is a matching of size at least $\frac{n-1}{3}$ 
in $G$ where $y,z$ are unmatched and no artificial edges are
used.  Finally add $(y,z)$ to this matching; this gives a matching
in $G$ of size $\frac{n+1}{3}$ as desired.
Since the 2-connected components can be found in linear time,
and matching $M_x$ is obtained by computing a matching in
each 2-connected component, the overall run-time is linear.
\qed
\end{proof}

\section{Higher minimum degree?}

Nishizeki and Baybars \cite{NB79} proved larger matching-bounds
for higher minimum degree, in particular 
3-connected planar graphs have matchings of size
at least $\frac{3}{7}n$ if the minimum degree is 4.
Their proof uses the Tutte-Berge-formula, so naturally one wonders
whether these matchings can be found in linear time.
This remains an open problem.  Lemma~\ref{lem:mindeg4} below
states that a slightly smaller matching {\em can} be found in
linear time.  This result is not new:
Applying
Baker's approximation scheme \cite{Baker94} with 
$\varepsilon=\frac{1}{15}$ will gives a matching of size at least
$\left(1-\frac{1}{15}\right)\frac{3}{7}n=\frac{6}{15}n$ in time
$O(8^{15} n)$.
Since $8^{15}$ is a constant (albeit large), this is linear time.
The interesting part of Lemma~\ref{lem:mindeg4} is hence not the
result, but its proof; the hope is that with minor modifications
it could be used to find a matching of size $\frac{3}{7}n$  in
3-connected planar graphs with minimum degree 4.


\begin{lemma}
\label{lem:mindeg4}
Let $G$ be a triangulated planar graph with $n>3$
where the outer-face consists of $\{x,y,z\}$
and all interior vertices have degree at least 4.  Then $G$ has 
\begin{itemize}
\item a matching $M_{x}$ of size at least $\frac{6n-15}{15}$ where $y,z$ are
unmatched.
\item a matching $M_{\emptyset}$ of size at least $\frac{6n-21}{15}$ where $x,y,z$ are unmatched.
\end{itemize}
They can be found in linear time.
\end{lemma}
\begin{proof} (Sketch, details can be found in the appendix.)
Let $P$ be a {\em Tutte-path}, which is a path such that any 
component $C$ of $G\setminus P$ has only three neighbours on $P$,
hence resides inside a {\em separating triangle} $T_C$ (a triangle
of $G$ that is not a face).  Furthermore, each component $C$
has a {\em representative} $\sigma(C)$, which is a vertex in $P$
such that no two components use the same representative.  For
each component $C$, let $C^+$ be the graph induced by $T_C\cup C$.
Recursively compute $M_\emptyset(C^+)$ and $M_{\sigma(C)}(C^+)$.

Now take every other edge of $P$, say $(v_{i-1},v_i)$ is such 
an edge.  If neither $v_{i-1}$ nor $v_i$ is used as representative,
then add $(v_{i-1},v_i)$ to $M$; this adds one matching-edge for two
vertices.  If $v_{i-1}$ and/or $v_i$ is used as representatives, say 
for components $C_{i-1}$ and/or $C_i$, then add the matchings of these
components.  Evaluating the constants exactly, one can verify that this
adds $\frac{6(|C_i|+|C_{i-1}|+2)}{15}$ matching-edges for the
$|C_i|+|C_{i-1}|+2$ vertices in $C_i\cup C_{i-1}\cup \{v_{i-1},v_i\}$.  
Repeating this for every other
edge of $P$ gives a matching of size $\frac{6|V|}{15}-c$
(where the error term $-c$ occurs if $P$ has an even number of edges).
This error-term is a bit too big for the desired bound, 
but one can argue (by re-arranging which matchings are used) that
it can be decreased to be at most $-\frac{3}{15}$.

To ensure that $y,z,$ and (maybe) $x$ are unmatched, path $P$ must be
chosen so that it begins at $x$ and ends with $(y,z)$.
Then use $P\setminus \{y,z\}$ or $P\setminus\{x,y,z\}$ as path in
the above arguments  and the desired bound follows.
\qed
\end{proof}

\section{Outlook}

This paper considered the problem of finding large matchings in
a planar graph with minimum degree 3.  It was known that any such
graph has a matching of size at least $\frac{n+2}{3}$, but all 
previous proofs were non-constructive and relied on generic
maximum-matching-finding algorithms to construct such a matching.
This paper re-proves the result (for varying levels of connectivity),
and the proofs naturally lead to linear-time algorithms to find 
the matching.

As for open problems,
a number of other matching bounds have been proved using the Tutte-Berge
formula, sometimes for graph classes that are not even planar:
\begin{itemize}
\item 
	As mentioned earlier, every 3-connected planar graph has a
	matching of size at least $\frac{3}{7}n$ if $n\geq 30$ and
	the minimum degree is 4, and of 
	size at least $\frac{9}{19}n$ if $n\geq 78$ and
	the minimum degree is 5 \cite{NB79}.
\item 
	Any 3-connected
	planar graphs have a matching of size $(2n-\ell_4)/4$,
	where $\ell_4$ is the number of leaves in the tree $\calT_4$ of
	4-connected components \cite{BDD+04}.    
\item Every Delauney 
	triangulation has a 
	matching of size $\lceil (n-1)/2 \rceil$ \cite{Dill90}.
\item Every $\Theta_6$-graph has a matching of size at least
	$\frac{3}{7}n$ \cite{BBI+18}.
\item Every maximal simple 1-planar graph has a matching of
	size at least $\frac{2}{5}n$, and there are also
	matching bounds for 1-planar graphs of minimum degree 3,4, or 5 
	\cite{BW19}.
\end{itemize}

For all these results, the following question remains open:
Can such matchings be bound in linear time?

\bibliographystyle{plain}
\bibliography{journal,full,gd,papers}


\begin{appendix}
\section{Triangulated graphs with minimum degree 4}

This appendix gives a proof of Lemma~\ref{lem:mindeg4},
i.e., a triangulated graph with minimum degree 4 has a matching 
of size at least $\frac{6}{15}n$ that can be found in linear time.
More precisely, for $X=\{y,z\}$ or $X=\{x,y,z\}$ the goal is
to find a matching $M_X$ of size $\frac{6(n-|X|)-3}{15}$ where vertices
in $X$ are unmatched.

\paragraph{Base case:}

The proof proceeds by induction on the number of separating triangles.
In the base case, $G$ has no separating triangles, and therefore is
4-connected.  In turns out that the induction step given below can
also cover the base case, so we will not write the argument here.

\paragraph{Tutte path background:}
The matching is found using a so-called {\em Tutte-path} $P$, which is a generalization
of a Hamiltonian path, 
exists in 2-connected planar graphs \cite{Tutte77} and can be found
in linear time \cite{BiedlKindermann}.   The following statement holds
for all 3-connected planar graphs (even if not triangulated), and needs
the notion of a {\em separating triplet}, i.e., a set $T$ of three vertices
such that $G\setminus T$ is not connected.

\begin{theorem}[\cite{BiedlKindermann}]
\label{thm:BK}
Let $G$ be a 3-connected
planar graph where the outer-face contains a vertex $x$
and an edge $(y,z)$ with $y\neq x\neq z$.  Then $G$ has a path $P$ that
begins with $x$ and ends with edge $(y,z)$ such that the following holds
for the connected components $\calC$ of $G\setminus P$:
\begin{itemize}
\item For any $C\in \calC$, there exists a separating triplet 
	$T_C=\{x_C,y_C,z_C\}$ of $G$ such that $C$ is a component of $G\setminus T_C$
	and the vertices of $T_C$ are on $P$.
\item There exists an injective assignment $\sigma(C): \calC\rightarrow P$
	of {\em representatives} such that $\sigma(C)\in T_C \setminus 
	\{x,y,z\}$.\
\end{itemize}
Furthermore, $P$ and the representatives can be found in time
$O(\sum_{f\in \calF(P)} \deg(f))$, where $\calF(P)$ is the set of 
interior faces of $G$ that contain a vertex of $P$.
\end{theorem}

Assume for the rest of the proof 
that one such path $P$ has been fixed. 
Note that vertices in $X$ are the first or last vertices of $P$,
so $P\setminus X$ is also a path and can be enumerated as
$\langle v_0,v_1,v_2,\dots,v_k\rangle$.  See also
Figure~\ref{fig:attach_bridges}.  Observe that $n\geq 4$ implies $|P|\geq 4$,
otherwise any interior vertex is in some component $C\in \calC$,
but $\sigma(C)\in P\setminus \{x,y,z\}$, which is then impossible.
So $P\setminus X$ is non-empty and $k\geq 0$.  On the other hand,
$\calC$ may well be empty (e.g.~in the base case when there are
no separating triangles).

If $\calC$ is non-empty, then
for each $C\in \calC$ let $T_C$ be the separating triangle that encloses $C$,
and let $C^+$ be the  subgraph
induced by $C\cup T_C$.
Since $C^+$ is triangulated, has fewer separating triangle, and all
its interior vertices have degree 4 or more, one can obtain
matchings $M_\emptyset(C^+)$ and $M_{v}(C^+)$ for $v\in T_C$
recursively.    Classify $C$ as one of two types:  $C$ is
{\em outward} if $6|C|\%15\in \{1,2,3\}$ and {\em inward} otherwise,
where $\%$ denotes the modulo-operator.  Usually an
inward component $C$ will use $M_\emptyset(C^+)$ while an outward
component will use $M_\sigma(C)(C^+)$, but some exceptions will
be made to this below.

\paragraph{Pre-processing:}
Before computing the matching, the assignment $\sigma(\cdot)$ of representatives
should get changed so that, whenever possible, an even-indexed
vertex is used.  More precisely,
	assume that some $C\in \calC$ satisfies the following:
	$\sigma(C)$ is $v_{2s+1}$, 
	some other vertex of $T_C$ is $v_{2t}$, 
	and there is no $C'\in \calC$ with $\sigma(C')=v_{2t}$.  Then
	change the system of representatives and set $\sigma(C):=v_{2t}$.
	For example, in Figure~\ref{fig:attach_bridges}(a), one can reassign
	$C_6$ to use $v_6$, rather than $v_7$, as representative.

\begin{figure}[ht]
\subcaptionbox{}{\includegraphics[width=0.5\linewidth,page=1]{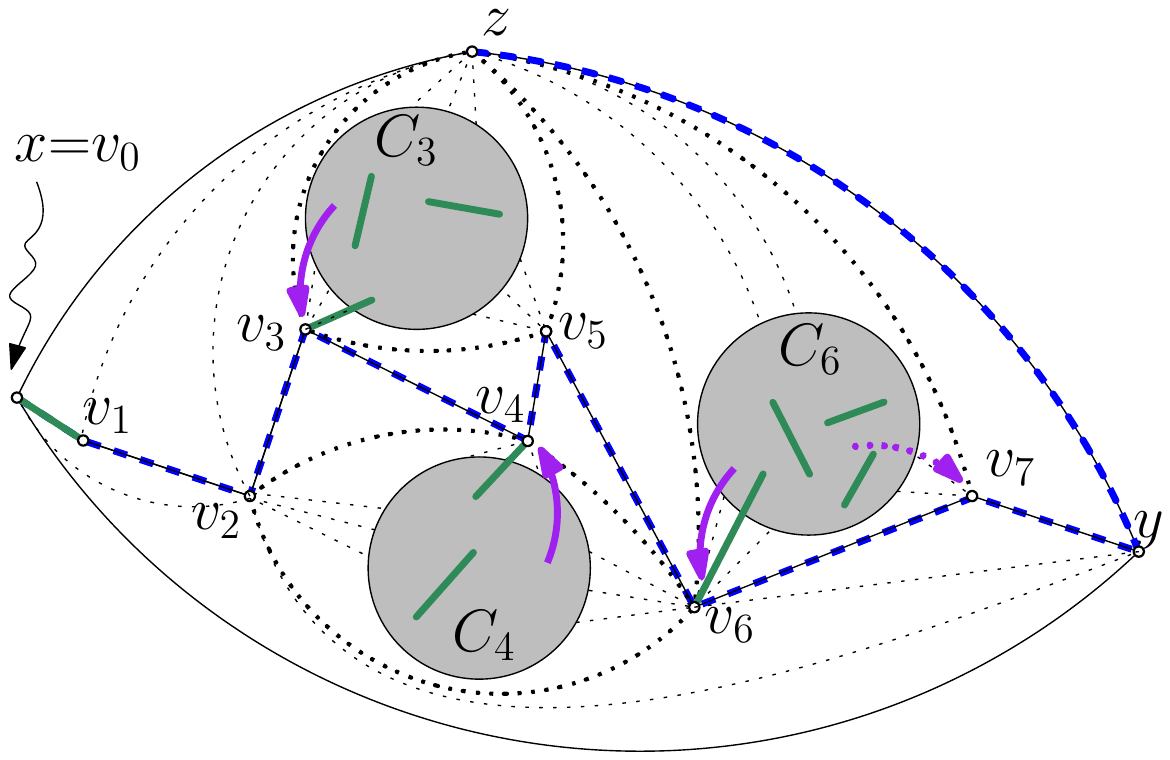}}
\subcaptionbox{}{\includegraphics[width=0.5\linewidth,page=2]{attach_bridges.pdf}}
\caption{Finding $M_x$. Path $P$ is blue dashed, the representatives are
purple arrows.  (a) $k$ is odd, so use
$M_{0,1}, M_{2,3}, M_{4,5}$ and $M_{6,7}$.
(b) $k$ is even and $C_2$ is non-empty, so use $M_{0,1}$,
$M_2$, $M_{3,4}$ and $M_{5,6}$. 
}
\label{fig:attach_bridges}
\end{figure}

\paragraph{Matchings $M_i$ and $M_{i-1,i}$:}
For $i=0,\dots,k$, let $C_i$ be the component in $\calC$ for which $\sigma(C_i)=v_i$.
If there is no such component then set $C_i:=\emptyset$ (an {\em empty} component).
Also set $V_i:=V(C_i)\cup \{v_i\}$.  Observe that $\bigcup_{i=0}^k V_i = V\setminus X$
since no component uses a vertex of $X\subseteq \{x,y,z\}$ as a representative.
Since the goal is to find a matching $M$ of size $\frac{6(n-|X|)-3}{15}$,
it suffices to find for $i=0,\dots,k$ a matching in $V_i$ of size $\frac{6|V_i|}{15}$.
This is not always possible (e.g. if $C_i$ is empty), but the idea is to achieve this
bound on average.

For $i=0,\dots,k$, define matching $M_i$ as follows.
If $C_i$ is empty then set $M_i:=\emptyset$ and observe that 
$|M_i|=0=\frac{6|V_i|-6}{15}$.
If $C_i$ is inward, then set $M_i:=M_\emptyset(C_i^+)$, and
observe that 
$|M_i|\geq \lceil \frac{6|C^+_i|-21}{15} \rceil = \lceil \frac{6|C_i|-3}{15} \rceil$
by integrality.  But $6|C_i|\% 15 \neq 1,2,3$ by definition of ``inward'',
so actually $|M_i|\geq \frac{6|C_i|}{15}=\frac{6|V_i|-6}{15}$.
Finally if $C_i$ is outward, then set $M_i:= M_{v_i}(C_i^+)$. 
By integrality $|M_i| \geq \lceil \frac{6|C_i^+|-15}{15} \rceil = 
\lceil \frac{6|C_i|+3}{15} \rceil$.  Since $6|C_i|\% 15 \in \{1,2,3\}$, hence
$\lceil \frac{6|C_i|+3}{15} \rceil \geq \frac{6|C_i|+12}{15}=\frac{6|V_i|+6}{15}$.  
Summarizing,
\mbox{for $i\in \{0,\dots,k\}$, } 
\begin{equation}
\label{equ:Msingle}
|M_i|\geq 
\left\{ \begin{array}{ll}
\frac{6|V_i|+6}{15} & \mbox{if $C_i$ is outward} \\[1.2ex]
\frac{6|V_i|-6}{15} & \mbox{if $C_i$ is inward or empty} 
\end{array} \right.
\end{equation}
Furthermore, $v_i$ is unmatched in $M_i$ unless $C_i$ is outward.

For $i=1,\dots,k$, define a matching $M_{i-1,i}$ within $V_{i-1}\cup V_i$ as
follows:  If both $C_{i-1}$ and $C_i$ are empty or inward, then set $M_{i-1,i}=(v_{i-1},v_i)
\cup M_{i-1} \cup M_{i}$; this is a matching since $v_{i-1}$ and $v_i$ are unmatched
in $M_{i-1}$ and $M_i$.  In this case,
$|M_{i-1,i}|\geq 1+
\frac{6|V_{i-1}|-6}{15}
+\frac{6|V_{i}|-6}{15}
= \frac{6(|V_{i-1}|+|V_i|)+3}{15}$.  Otherwise,
set $M_{i-1,i}=M_{i-1}\cup M_i$ and observe that this has size at least
$\frac{6|V_{i-1}|-6}{15}+\frac{6|V_i|+6}{15}$ since at least one of $C_{i-1}$ and $C_i$
is outward.%
\footnote{Here is where replacing $\frac{6}{15}$ by $\frac{3}{7}$ would not work out:
The ``surplus'' in $M_i$ if $C_i$ is outward must be high enough to make up
for the ``deficit'' in $M_{i-1}$ if $C_i$ is empty.}
So \mbox{for $i=1,\dots,k$}
\begin{equation}
\label{equ:Mdouble}
|M_{i-1,i}|\geq 
\left\{ \begin{array}{ll}
\frac{6(|V_{i-1}|+|V_i|)+3}{15} & \mbox{if neither $C_{i-1}$ nor $C_i$ is outward} \\[1.2ex]
\frac{6(|V_{i-1}|+|V_i|)}{15} & \mbox{otherwise}
\end{array} \right.
\end{equation}

\paragraph{Combining matchings:}

The following cases show how to put these matchings together to obtain
the desired bound:

\medskip

\begin{description}
\itemsep +2pt
\item[Case 1:] $k$ is odd.  Then use 
$ M = M_{0,1} \cup M_{2,3} \cup \dots \cup M_{k-1,k}.$
By Inequality (\ref{equ:Mdouble}) then $|M|\geq \sum_{i=0}^k \frac{6|V_i|}{15}=\frac{6(n-|X|)}{15}$ as desired.

\item[Case 2:] $k$ is even and $C_{2s}$ is non-empty for some $s=0,\dots,k/2$.  Then use
$$ M = M_{0,1} \cup \dots \cup M_{2s-2,2s-1} \cup M_{v_{2s}}(C^+_{2s}) \cup M_{2s+1,2s+2} \cup \dots \cup M_{k-1,k}.$$
Note that regardless of whether $C_{2s}$ is inward or outward,
we have $|M_{v_{s}}(C^+_{2s})|\geq \frac{6|C_{2s}^+|-15}{15} = \frac{6|V_{2s}|-3}{15}$.
Combining this with Inequality (\ref{equ:Mdouble}) applied to all other
parts of $M$ shows that $|M|\geq \sum_{i=0}^k \frac{6|V_i|}{15}-\frac{3}{15}=\frac{6(n-|X|)-3}{15}$ as desired.

\item[Case 3:] $k$ is even, and $C_{2s+1}$ and $C_{2s+2}$ are both
empty or inward for some $0\leq s \leq k/2-1$. Set
$$ M = M_{0,1} \cup \dots \cup M_{2s-2,2s-1} \cup M_{2s} \cup M_{2s+1,2s+2} \cup \dots \cup M_{k-1,k}.$$
(This is almost the same matching as in the previous case, but the analysis is different.)
The stronger bound in Inequality (\ref{equ:Mdouble}) applies to $M_{2s+1,2s+2}$, so
$$ |M_{2s}|+|M_{2s-2,2s-1}|\geq
\frac{6|V_{2s}|-6}{15}+\frac{6(|V_{2s+1}|+|V_{2s+2}|)+3}{15}.$$
Combining this with Inequality (\ref{equ:Mdouble}) applied to all other
parts of $M$ shows that $|M|\geq \sum_{i=0}^k \frac{6|V_i|}{15}-\frac{3}{15}=\frac{6(n-|X|)-3}{15}$ as desired.

\item[Case 4:] 
$k$ is even, and for some $0\leq s < t\leq k/2$ 
edge $(v_{2s},v_{2t})$ exists and $C_{2s}$ and $C_{2t}$ are empty.  Then use
\begin{eqnarray*}
M & = & M_{0,1} \cup \dots \cup M_{2s-2,2s-1} \quad \cup \quad 
	(v_{2s},v_{2t})  \cup M_{2s+1} \\
	& & \cup M_{2s+2,2s+3} \cup \dots \cup M_{2t-2,2t-1} \quad \cup \quad 
	M_{2t+1,2t+2} \cup \dots \cup M_{k-1,k}.
\end{eqnarray*}
See also Figure~\ref{fig:attach_bridges_2}(b).
Note that $(v_{2s},v_{2t}) \cup M_{2s+1}$ resides within $V_{2s}\cup V_{2s+1}\cup V_{2t}$, and 
$|V_{2s}|=|V_{2t}|=1$ since $C_{2s}$ and $C_{2t}$ are empty, so
$$1{+}|M_{2s+1}| \geq \frac{6(|V_{2s}|{+}|V_{2t}|){+}3}{15} {+} \frac{6|V_{2s+1}|-6}{15} = \frac{6(|V_{2s}| {+} |V_{2s+1}|{+}|V_{2t}|)-3}{15}.$$
Combining this with Inequality (\ref{equ:Mdouble}) applied to all
other parts of $M$ shows that $M$ has 
size at least $\sum_{i=0}^k \frac{6|V_i|}{15}-\frac{3}{15}=\frac{6(n-|X|)-3}{15}$ as desired.
\end{description}

\begin{figure}[ht]
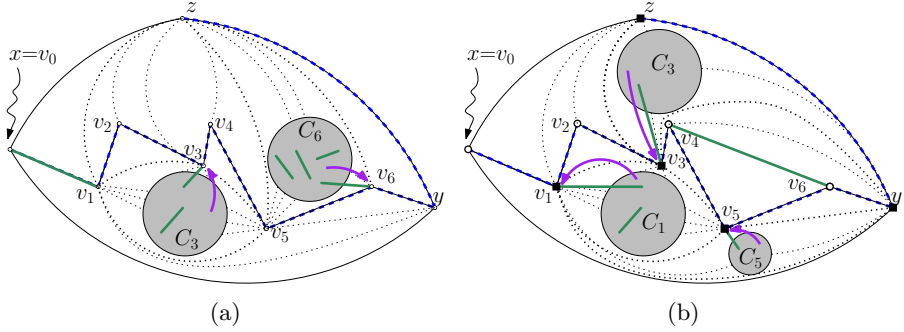

\hspace*{\fill}
\subcaptionbox{}{\includegraphics[width=0.48\linewidth,page=3]{attach_bridges.pdf}}
\hspace*{\fill}
\subcaptionbox{}{\includegraphics[width=0.48\linewidth,page=4]{attach_bridges.pdf}}
\hspace*{\fill}
\caption{More methods of obtaining $M$.
(a) $k$ is even and $C_0,C_1$ are empty, so
use $M_{0,1}\cup M_2 \cup M_{3,4} \cup M_{5,6}$.
(b) $k$ is even, edge $(v_4,v_6)$ exists and $C_4,C_6$ are empty,
so use $M_{0,1}\cup M_{2,3} \cup  (v_4,v_6) \cup M_5$.   
}
\label{fig:attach_bridges_2}
\end{figure}

\paragraph{Correctness:}
It remains to show that one of the above cases must apply.  Assume for contradiction 
that none of them applies, so in particular $k$ is even (by Case 1), 
$C_0,C_2,\dots,C_k$ are empty (by Case 2) and $C_1,C_3,\dots,C_{k-1}$
are outward (by Case 3 and since $C_0,C_2,\dots,C_k$ are empty).
Construct an auxiliary graph $H$ as follows (this step is only needed to argue 
correctness, not for the algorithm).  Vertex set $A$ (shown with
black squares in Figure~\ref{fig:attach_bridges_2}(b)) consists of 
$X$ and $v_1,v_3,\dots,v_{k-1}$,
a total of $\frac{k}{2}+|X|$ vertices.  Vertex set $B$ (shown with circles)
consists of $v_0,v_2,\dots,v_k$ and vertices $c_1,c_3,\dots,c_{k-1}$, where 
$c_i$ represents the (non-empty, outward) component $C_i\in \calC$.   
So $|B|=k+1$.     Note that $B$ includes $x=v_0$
if (as in Figure~\ref{fig:attach_bridges_2}(b)) the goal is to find matching
$M_x$.

Add the following edges to $H$.  First, add any edge from $A$ to $B$
that also existed in $G$.  Also connect $c_i$ (for $i=1,3,\dots,k-1$) to each
vertex of $T_{C_i}$; due to the pre-processing  vertices in $T_{C_i}$
are in $X$ or have an odd index, so $c_i$ has exactly three neighbours and
they are in $A$.    

By construction $H$ is bipartite and a minor of planar graph $G$. 
Since also $|H|\geq 3$, it has at most $2|H|-4=2(\frac{3}{2}k+|X|+1)-4 
= 3k+2|X|-2$ edges.  
Observe that $\deg_H(v_{2s})\geq 4$ for
all even-indexed vertices $v_{2s}$, because $v_{2s}$ has at least four
neighbours in $G$, none of these can be in a component $C_{2t}$ ($C_{2t}$
is empty), none of these can be in a component $C_{2t+1}$ 
(else $v_{2t}\in T_{C_{2t+1}}$, but by the pre-processing $T_{C_{2t+1}}\subseteq A$)
and none of them can be an even-index vertex $v_{2t}$
(else Case 4 applies).  So all neighbours of $v_{2s}$ in $G$ are in
$X\cup \{v_1,v_3,\dots,v_{k-1}\}=A$ and also neighbours in $H$.  It follows
that $H$ has at least $\sum_{v\in B} \deg(v)\geq 4(\frac{k}{2}+1)+3\frac{k}{2} = \frac{7}{2}k+4$
edges.  Thus $\frac{7}{2}k+4\leq 3k+2|X|-2$, which implies $|X|=3$ and $k=0$.
But then $|A|=3$, making $\deg_H(v_0)\geq 4$ impossible.  Contradiction,
so one of the cases must apply.

\paragraph{Run-time:}

Some algorithmic details remain to be filled in.
Before starting any recursions,
compute the tree $\calT_4$ of 4-connected components; this can be done in linear time
\cite{Kant-IJCGA97}.    Root $\calT_4$ at the component containing the
outer-face, and let every node store that size of the subgraph of its
descendants.  With this, one can later determine in constant time whether 
$6|C|\% 15 \in \{1,2,3\}$ for any $C\in \calC$, because $C$ is one component
of $G\setminus T_C$ for a separating triangle $T_C$ and hence corresponds to
a node of $\calT_4$.

Finding $P$ and the representatives $\sigma$
takes time $O(\sum_{f\in \calF(P)} \deg(f))$, where
$\calF(P)$ are the interior faces of $G$ that are incident to vertices of $P$     \cite{BiedlKindermann}.
Next, determine triangle $T_C$ for each $C\in \calC$ (this ``falls out of''
the algorithm to find $P$, but can also be re-computed by scanning the clockwise
order of edges around each vertex in $P$).  Then do the pre-processing to find
the final assignment of representatives; note that at this point the indices
of vertices in $P$ are fixed and so this can be done in $O(1)$ time per component.
Finally find the graph $G[P]$ induced by $P$ and check whether it has any edges
between two even-indexed vertices.  
Since one can also look up the type of each $C\in \calC$ in constant time,
one can hence determine which case applies in $O(|P|)$ time.
All these operations together take time
$O(|P|+\sum_{v\in P} \deg(v))$, which is $O(\sum_{f\in \calF(P)} \deg(f))$ and
can hence be viewed as overhead to computing $P$.

Now recursively compute for each $C\in \calC$ the required matching; 
this takes time $O(|\calF(C)|)$ time, where $\calF(C)$
is the set of interior faces in $C^+$.   Note that $\calF(C)$ and $\calF(P)$ overlap
in faces that are incident to $T_C$, but similarly as in \cite{BiedlKindermann}
it is possible to re-use the information from the computation for $P$ in the
computation of the path $P_C$ in $C^+$ (which is the first step in the
computation of the matching in $C^+$).  So overall this is linear time.
This finishes the proof of Lemma~\ref{lem:mindeg4}.
\end{appendix}

\end{document}